\documentclass{article}
\usepackage[utf8]{inputenc}
\usepackage{etex}
\usepackage{graphicx}
\usepackage{amscd}
\usepackage{amsfonts}
\usepackage{amssymb}
\usepackage{amsmath}
\usepackage[amsmath,amsthm,thmmarks]{ntheorem}
\usepackage{mathtools}
\usepackage{curves}
\usepackage{eucal}
\usepackage{epsfig}
\usepackage{enumerate}
\usepackage{extarrows}
\usepackage{fullpage}
\usepackage{hyperref}
\usepackage[utf8]{inputenc}
\usepackage{ifthen}
\usepackage{latexsym}
\usepackage[usenames,dvipsnames]{pstricks}
\usepackage{pst-grad} 
\usepackage{pst-plot} 
\usepackage[all]{xy}
\usepackage{tikz}
\usepackage{verbatim}
\usepackage{xargs}

\newcounter{comments}
\usepackage{comment}

\newcommand{\new}[2][]{#2}

\DeclareFontFamily{OT1}{rsfs}{}
\DeclareFontShape{OT1}{rsfs}{n}{it}{<-> rsfs10}{}
\DeclareMathAlphabet{\mathscr}{OT1}{rsfs}{n}{it}

\DeclareMathOperator{\E}{E}
\DeclareMathOperator{\fix}{fix}

\DeclareMathOperator{\ord}{ord}

\newtheorem{ex}{Exercise}

\newtheorem{dfn}{Definition}[section]

\newtheorem{cor}[dfn]{Corollary}
\newtheorem{eg}[dfn]{Example}
\newtheorem{nota}[dfn]{Notation}
\newtheorem{lem}[dfn]{Lemma}

\newtheorem{remark}[dfn]{Remark}

\newtheorem{thm}[dfn]{Theorem}
\newlength{\remaining}

\newcommand{\bb}{\mathbb}

\newcommand{\be}[1]{\begin{equation}#1
\end{equation}}
\newcommand{\bea}[1]{\begin{eqnarray}#1
\end{eqnarray}}
\newcommand{\bean}[1]{\begin{eqnarray*}#1
\end{eqnarray*}}

\newcommand{\bpm}[1]{\begin{pmatrix}#1
\end{pmatrix}}

\newcommand{\disp}{\displaystyle}
\newcommand{\ee}{\end{equation}}
\newcommand{\eea}{\end{eqnarray}}
\newcommand{\eean}{\end{eqnarray*}}
\newcommand{\een}{\end{equation*}}
\newcommand{\epm}{\end{pmatrix}}

\newcommand{\f}{\mathfrak}

\newcommand{\la}{\left|}

\newcommand{\lr}{\left(}

\newcommand{\lt}{\left\langle}

\newcommand{\ra}{\right|}

\newcommand{\rr}{\right)}

\newcommand{\rt}{\right\rangle}

\newcommand{\x}{\xymatrix@1@=50pt@M=4pt@L=3pt}

\title{Stability of Linear Boolean Networks}
\author{{Karthik Chandrasekhar}\\
Rashtreeya Vidyalaya University\\
\and
{Claus Kadelka}\\
Iowa State University\\
\and
{Reinhard Laubenbacher}\\
University of Florida\\
\and
{David Murrugarra}\\
University of Kentucky\\
}

\begin{document}

\maketitle

\begin{abstract}
Stability is an important characteristic of network models that has implications for other desirable aspects such as controllability. The stability of a Boolean network (BN) depends on various factors, such as the topology of its wiring diagram and the type of the functions describing its dynamics. Linear Boolean networks can be completely described by their wiring diagram, and therefore the structure of linear networks plays a prominent role in determining their stability. In this paper, we study the stability of linear Boolean networks by computing Derrida curves and quantifying the number of attractors and cycle lengths imposed by their network topologies. Derrida curves are commonly used to measure the stability of Boolean networks and several parameters such as the average in-degree $K$ and the output bias $p$ can indicate if a network is stable, critical, or unstable. For random unbiased Boolean networks there is a critical connectivity value $K_c=2$ such that if $K<K_c$ networks operate in the ordered regime, and if $K>K_c$ networks operate in the chaotic regime. Here, we show that for linear networks, which are the least canalizing and most unstable, the phase transition from order to chaos already happens at an average in-degree of $K_c=1$. Consistently, we also show that unstable networks exhibit a large number of attractors with very long limit cycles while stable and critical networks exhibit fewer attractors with shorter limit cycles. Additionally, we present theoretical results to quantify important dynamical properties of linear networks. First, we present a formula for the proportion of attractor states in linear systems. Second, we show that the expected number of fixed points in linear systems is 2, while general Boolean networks possess on average one fixed point. Third, we present a formula to quantify the number of bijective linear Boolean networks and provide a lower bound for the percentage of this type of network.

\end{abstract}
\section{Introduction}

Boolean networks (BNs) are popular models used in biology and engineering due to their intuitive formalism, and their ability to capture important dynamical features of biochemical networks without the need for estimating precise kinetic rate constants~\cite{thomas1990biological,veliz2022building}. A Boolean network on $n$ variables is described by a ``wiring diagram," a directed graph on $n$ nodes, together with a Boolean coordinate function attached to each graph node. Iterative application of these functions to binary strings of length $n$ generates a dynamical system for the Boolean network. Its dynamics can be described in terms of another directed graph, its state space graph. \new[It has]{This graph contains} all $2^n$ binary strings of length $n$ as nodes, with directed edges between states capturing the action of the Boolean functions. In general, the state space graph can only be computed by exhaustive simulation of the Boolean network, but can in special cases be optimized by utilizing properties of the wiring diagram or the regulatory functions. In general, there is no mathematical theory for complete characterization of the dynamics of all Boolean networks. For some classes of Boolean networks, theoretical results to obtain important information on their dynamics \new[has]{have} been established. In particular, for linear Boolean networks, a mathematical framework is available for obtaining the number of attractors and the structure of their attractor basins~\cite{RHT}. Moreover, for linear systems, one can derive a generating function that provides the number of cycles of a given length~\cite{RHT}. Similarly, for conjunctive networks (that is, Boolean networks whose Boolean functions are all AND rules) with strongly connected wiring diagrams, \new[the]{a} formula \new[to obtain the number of attractors was given]{for the number of their attractors exists}~\cite{jarrah2010dynamics}. For general conjunctive networks, a sharp lower bound has been provided~\cite{jarrah2010dynamics}.

Boolean networks are discrete-time dynamical systems where the state of each variable at the next time step is determined by a Boolean function over a subset of the system variables. Attractors are sets of states in which the system will be trapped as it evolves. Given a Boolean network model, one commonly associates \new[the important states of the system with the attractors of the model]{its attractors with important states of the system}. For instance, when modeling gene regulatory networks, attactors are usually associated with the possible phenotypes of the cell~\cite{kauffman1969homeostasis,huang2009cancer}.
\new[For another example,]{Or, in cancer models,} attractors \new[in cancer modeling]{} might represent a differentiated cell type~\cite{kauffman1969metabolic} or a cellular state such as apoptosis, proliferation, or cell senescence~\cite{huang1999gene, plaugher2021modeling, plaugher2022uncovering}. Knowing the number of attractors of Boolean networks is very important as \new[this]{it} is related to the stability and controllability of networks~\cite{kauffman2004genetic,murrugarra2011regulatory}. 

\new{Linear systems often serve as useful approximations in the analysis of continuous-variables models such those based in ordinary differential equations~\cite{strogatz2018nonlinear}. For discrete systems, linearization approaches have not been used. However, any Boolean network can be represented as a linear system in a higher dimensional space \cite{shmulevich2002probabilistic}. Moreover, in recent years, the semi-tensor product representation of Boolean networks have been used for the analysis and control of Boolean networks \cite{cheng2010analysis,luo2014controllability,chen2021model}. Indeed, the semi-tensor representation is a type of linear representation. Thus, obtaining new results for linear systems is still an important endeavor and many of these results could be useful for the analysis of nonlinear Boolean networks.}

In this paper, we implement the approach in~\cite{RHT} to compute the attractor distributions of linear systems in terms of the average number of attractors of a given length.
One advantage of the approach in~\cite{RHT} is that it allows us to calculate the number of attractors without identifying the actual attractor states which makes the task of quantifying the attractor distribution more efficient.

For several classes of Boolean functions and their networks, theoretical results to obtain statistics about their dynamics have been established. Even though linear systems are the simplest class of Boolean networks and several results are available to study their dynamics, important information remains unknown. This paper contributes additional results of this kind. For instance, we provide information about the expected number of fixed points. 
For conjunctive networks, the expected number of fixed points is 2~\cite{jarrah2010dynamics}. The average number of fixed points of random Boolean networks is 1~\cite{samuelsson2003superpolynomial,borriello2021basis}. Here we show that the average number of fixed points of linear systems is 2.

\new{This paper is structure in the following way.} In Section~\ref{sec:background} we \new[provide the definition of]{define} Boolean networks and \new{provide }other background \new{material}. In Section~\ref{sec:derrida} we present a formula to compute Derrida curves for linear systems along with plots of Derrida curves, attractor distribution, and size of their strongly connected components for networks with \new[Poisson, scale free, and constant]{constant and scale-free} in-degree distributions.
In Section~\ref{sec:attr_states} we present a formula for the proportion of attractor states in linear systems\new[.
In Section~\ref{sec:fps} we provide a]{, followed by a} formula for the average number of fixed points in linear systems \new{in Section~\ref{sec:fps}}.
In Section~\ref{sec:dim_fps} we \new[provide]{derive} a formula to quantify the number of bijective linear Boolean networks. Furthermore, we provide an estimate for the percentage of linear maps that are invertible. \new{Finally, we discuss our main findings and provide some conclusions in Section~\ref{sec:discussion}.}

\section{Background}\label{sec:background}
Boolean networks are dynamical systems that are discrete in time and state\new[ variables]{}. Specifically, consider a collection $x_1, \ldots , x_n$ of
variables, each of which can take on values in $\{0,1\}$. Then,
a (synchronously updated) Boolean network in the variables $x_1, \ldots , x_n$ is a function
$F= (f_1,\dots,f_n):\{0,1\}^n\to\{0,1\}^n$,
where each coordinate function $f_i$ is a discrete function on a subset of $\{x_1,\dots,x_n\}$, which represents how the future value of the $i$-th variable depends on the present values of the variables. When $F$ is a linear function \new{ (that is, all $f_i, i=1,\ldots,n$ are linear functions)}, then $F$ can be represented by a matrix. Thus, in the linear case, $F(x) = Mx$ for some \new{$n\times n$-}matrix $M$ over the field $\bb F_2$. For example consider the following:
\new{
\begin{eg}
For the case $n=3$, let
\be{F(x)=Mx}
where $\disp M=\bpm{1&1&0\\0&0&1\\0&0&1}$.
Now, if $\disp x=\bpm{1\\1\\1}$ we have
$F(x)=\bpm{0\\1\\1}$.
In general for the above matrix $M$,
\be{F\lr\bpm{x_1\\x_2\\x_3}\rr=\bpm{f_1(x)\\f_2(x)\\f_3(x)}=\bpm{x_1+x_2\\x_3\\x_3}\label{func_indiv}}
\label{linear_func_eg}
\end{eg}
}
Given a Boolean network $F = (f_1,\dots,f_n)$, its \emph{wiring diagram} $\mathcal{W}$ is defined to be the directed graph with $n$ nodes 
$x_1, \ldots , x_n$ associated to $F$, such that there is a directed edge in $\mathcal{W}$ from 
$x_j$ to $x_i$ if \new[$x_j$ appears in $f_i$, that is,]{and only if} the value of $f_i$ depends on $x_j$.

\new{
\begin{eg}
In the function $F$ described in equation \eqref{func_indiv}, the dependencies are simply that $f_1$ depends on $x_1$ and $x_2$, and that $f_2$ and $f_3$ depend on $x_3$ whereas no other dependencies are present. Therefore the wiring diagram is as in Figure \ref{fig:wiring}.
\begin{figure}
    \centering
    \includegraphics[height=10pc,width=3pc]{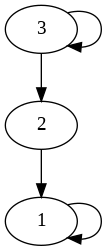}
    \caption{Wiring diagram for F(x) in Example \ref{linear_func_eg}}
    \label{fig:wiring}
\end{figure}
\end{eg}
}

The dynamics of discrete networks are given by the difference equation $x(t+1)= F(x(t))$; that is, the dynamics are generated by iteration of $F$. More precisely, the dynamics of $F$ are given by the \emph{state space graph} $S$, which has vertices
$\{0,1\}^n$ and an edge from $x\in \{0,1\}^n$ to $y\in \{0,1\}^n$ if and only 
if $y = F(x)$. \emph{Attractors} are terminal strongly connected components in the state space graph. 
Attractors are usually classified as either \emph{fixed points} or \emph{limit cycles}. 
Attractors of Boolean networks typically represent important outcomes. For example in a Boolean gene regulatory network model, attractors represent biological phenotypes.

\section{Stability of linear Boolean networks}\label{sec:derrida}

Linear Boolean functions can be considered as extreme in multiple ways. To know the output of a linear function one always needs to know all of its inputs. This is not true for any other function.  A consequence is that linear functions are the only functions with a canalizing strength of zero~\cite{kadelka2020collectively}, or zero input redundancy~\cite{gates2021effective}. Another frequently used stability measure, the average sensitivity of a Boolean function describes the sensitivity of a Boolean function to a single perturbation in \new{one of} its inputs~\cite{shmulevich2004activities}. For a linear function, the average sensitivity is always $1$ because a single change in its inputs flips the output. The average $c$-sensitivity of a linear function, a generalization of the average sensitivity (see~\cite{kadelka2017influence}), which can be thought of as the probability that the Boolean function output differs when exactly $c$ inputs differ, is $1$ if $c$ is odd and $0$ if $c$ is even. Using formulas described in~\cite{kadelka2017influence}, we can therefore very easily compute the Derrida plot of linear Boolean networks~\cite{derrida1986evolution}. This plot is frequently used to assess the stability of a Boolean network to perturbations. If a small perturbation becomes smaller on average after one update according to the Boolean rules, the network is in the \emph{stable} (or \emph{ordered}) {regime}. If the perturbation on average increases, it is in the \emph{chaotic regime}, and the small range in between where the perturbation on average remains of similar size is the \emph{critical regime}. Interestingly, most biological systems seem to operate in this critical regime~\cite{balleza2008critical,daniels2018criticality,kadelka2020meta}.

For random Boolean networks with output bias $p$ (the probability of having a 1 in the truth table) and average in-degree $k$, the phase transition from ordered to chaotic dynamics (i.e., the critical edge) occurs at $2kp(1-p) = 1$~\cite{derrida1986random,luque1997phase}. Assuming an unbiased selection of functions, that is $p=0.5$, this implies that random Boolean networks with an average degree of $k=2$ are critical. Random networks with $k<2$ exhibit on average ordered dynamics, while random networks with $k>2$ are on average unstable. Further, for networks governed by canalizing functions, e.g. most biological Boolean network models, the phase transition occurs at an average in-degree substantially greater than $k=2$~\cite{manicka2022effective}. Lastly, for nested canalizing functions, the average sensitivity is always 1 (i.e., nested canalizing networks always operate on average at the critical edge), irrespective of the average in-degree~\cite{kadelka2017influence}. We will now show that for linear networks, which are the least canalizing and most unstable, the phase transition from order to chaos already happens at an average in-degree of 1.

\begin{lem}\label{c_sens_linear}
The normalized average $c$-sensitivity of a Boolean linear function $f$ is $$q_c(f) = \begin{cases} 1 & \text{if $c$ is odd,}\\0 & \text{if $c$ is even}.\end{cases}$$
\end{lem}

\begin{thm}\label{thm_derrida_linear}
The Derrida value of a synchronously updated linear Boolean network $F=(f_1,\ldots,f_N)$ with in-degrees $n_1,\ldots,n_N$ can be expressed as a weighted sum of the normalized average $c$-sensitivities of its update functions,
\begin{equation}\label{derrida_thm}
D(F,m) := \mathbb{E}\Big[d\big(F(\mathbf x),F(\mathbf y)\big) \ \big| \ d(\mathbf x,\mathbf y) = m\Big] =  \sum_{i=1}^N\mathbb{P}\left(f_i(\mathbf x) \neq f_i(\mathbf y) \ \big| \ d(\mathbf x,\mathbf y) =m \right) = \sum_{i=1}^N \sum_{\substack{c=1\\ c\text{ odd}}}^{m}\ H_{N,m,n_i}(c),
\end{equation}
where
\[H_{N,m,n_i}(c) = \frac{\binom mc \binom{N-m}{n_i-c}}{\binom N{n_i}} = \frac{\binom {n_i}c \binom{N-n_i}{m-c}}{\binom N{m}}.\]
denotes the hypergeometric probability mass function.
\end{thm}
\begin{proof}
From~\cite[Theorem 4.3]{kadelka2017influence}, we have for any synchronously updated Boolean network $$D(F,m) = \sum_{i=1}^N \sum_{c=0}^{m}\ H_{N,m,n_i}(c) s_c(f_i),$$ where $s_c(f_i)$ is the normalized average $c$-sensitivity of the update function $f_i$. Plugging in the simple form of $s_c(f_i)$ for linear functions (Lemma~\ref{c_sens_linear}) completes the proof.
\end{proof}

\begin{remark}
For computation purposes, one can even aggregate all linear functions with the same number of inputs and compute the Derrida values faster using
\begin{equation}\label{derrida_formula_fast}
D(F,m) = \sum_{n\in \mathbf n(F)} w_n \sum_{\substack{c=1\\ c\text{ odd}}}^{m}\ H_{N,m,n}(c),
\end{equation}
where $n(F)$ is the set of all unique in-degrees in the linear functions governing $F$ and $w_n$ is the corresponding distribution, i.e., $\sum_{n\in \mathbf n(F)} w_n= 1$. 

Moreover, assuming $\max \mathbf n << N$  (which is true e.g. for large networks with a Poisson distributed in-degree distribution), one can use well-known approximations for the hypergeometric probability distribution to further speed up the computation.
\end{remark}

We computed Derrida plots for two types of random linear Boolean networks:
\begin{enumerate}[(i)]
    \item networks of size $N=20$ with a fixed in-degree $k$ where each update function $f_i, i=1,\ldots, 20$ has the same number of inputs.
    \item scale-free networks of size $N=20$ whose out-degree distribution follows a power law. In these networks, the probability that a node regulates $l=0,\ldots, 20$ nodes is $P_{\text{out}}(l) = C_1l^{-\gamma}$. Note that the in-degree of a scale-free network is Poisson distributed~\cite{aldana2007robustness}. That is, the probability that a node is regulated by $k=0,\ldots,20$ nodes is $P_{\text{in}}(k) = C_2e^{K}\frac{K^k}{k!}$ where $K$ represents the average degree of the network and is determined by the scale-free parameter $\gamma$. The constants $C_1,C_2 >1$ are needed because both distributions are truncated at $N=20$.
\end{enumerate}
The former network model has been extensively studied due to its simplicity and straightforward implementation~\cite{derrida1986phase,kauffman1993origins,aldana2003boolean,samuelsson2003superpolynomial}. However, most biological networks exhibit scale-freeness and are thus much better modeled by the latter model~\cite{leclerc2008survival,barabasi1999emergence,newman2005power}. We therefore investigated both types of network models. We considered fixed in-degrees of $k=1,2,$ and $3$ as well as corresponding scale-free parameters $\gamma = 2.41, 1.67,$ and $1.33$. This ensures that both the fixed in-degree network and the corresponding scale-free network have the same average degree.

\begin{figure}
    \centering
    \includegraphics{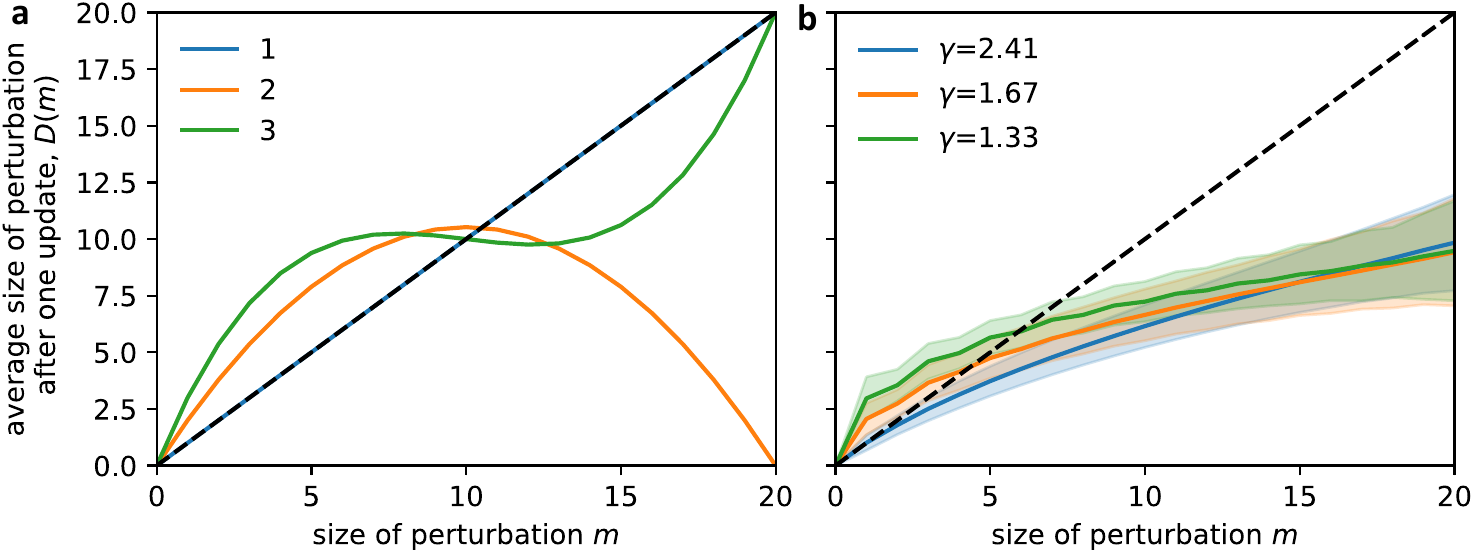}
    \caption{Derrida plot for linear networks with $N=20$ nodes and (a) fixed in-degree distribution of $k=1,2,3$ and (b) scale-free out-degree distribution with a \new{fixed} parameter $\gamma$ chosen to match the average degrees of the networks in (a). \new{That is, lines with the same color in (a) and (b) show results from linear networks with the same average degree.} Each curve is averaged across $50$ random networks. The shaded area signifies the standard deviation. A black dashed line, which coincides with the Derrida curve for linear networks with fixed in-degree of 1 in (a), highlights the critical threshold.}
    \label{fig:derrida}
\end{figure}

Linear functions give rise to networks with very unstable dynamics. Only linear networks with a fixed in-degree of $k=1$ exhibit critical dynamics~(Figure~\ref{fig:derrida}a). At a fixed in-degree of $2$, which for a random Boolean network yields critical behavior~\cite{derrida1986random,luque1997phase}, the corresponding linear network has already chaotic dynamics. The same is true for scale-free linear networks: networks with $\gamma=2.41$, corresponding to an average degree of $1$, operate at the critical edge, while networks with a higher average degree (i.e., $\gamma < 2.41$) are unstable~(Figure~\ref{fig:derrida}b). 

The Derrida value $D(F,1)$ depends only on the average degree. When the size of the perturbation is larger, the network topology matters. Linear networks with fixed even in-degree exhibit strange, non-monotonic Derrida curves. This is because the normalized average $c$-sensitivity for linear functions is non-monotonic. The Derrida curves for scale-free networks, on the other hand, are monotonically increasing and converge at $D(F,N) = N/2$. 

\begin{figure}
    \centering
    \includegraphics{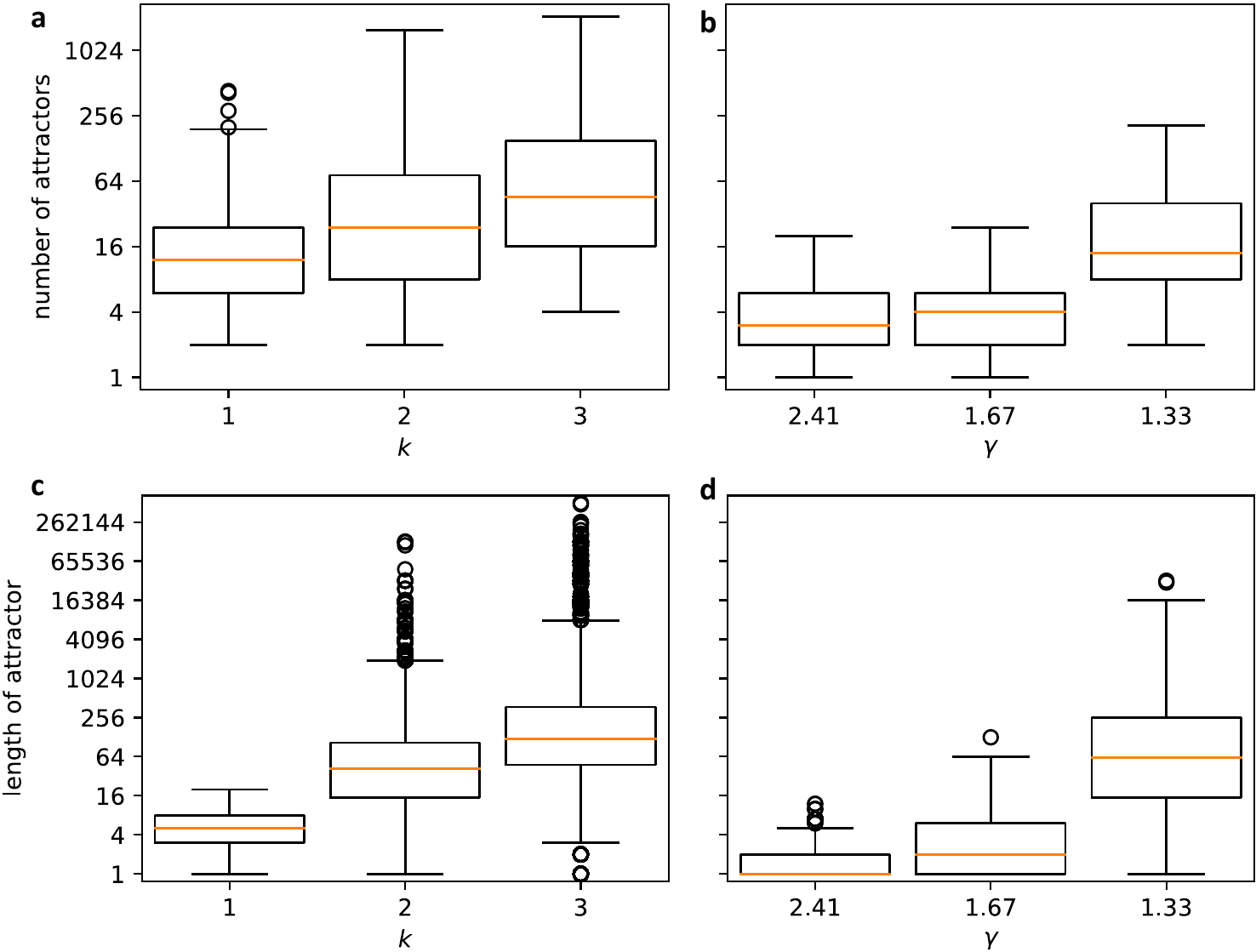}
    \caption{(a,b) Number and (c,d) length of attractors for linear networks with $N=20$ nodes and (a,c) a fixed in-degree $k$, and (b,d) scale-free out-degree distribution with a parameter $\gamma$ chosen to match the average degrees of the networks in (a). For each parameter value, $50$ random networks were generated. Orange lines depict the median, each box extends across the interquartile range (IQR), whiskers extend to the lowest data point still within 1.5 IQR of the lower quartile, and the highest data point still within 1.5 IQR of the upper quartile, and black circles show outliers. \new{The length of the whiskers was computed after log2-transformation of the data.}}
    \label{fig:attractors}
\end{figure}

Another characteristic of unstable networks is the presence of many and long attractors. For general Boolean networks, the computation of all attractors is time-consuming. However, for linear networks whose dynamics are fully described by their wiring diagram (or adjacency matrix), there exists a very fast method~\cite{RHT}, which we used here to compute length and number of all attractors for the networks already analyzed above. As the average in-degree of the networks increases, the number and length of attractors generally increases as well, a sign of increasingly chaotic dynamics~(Figure~\ref{fig:attractors}). Interestingly, scale-free networks possess both fewer and on average shorter attractors than networks with the same average degree but fixed in-degree. This may be due to the fact that the scale-free networks can possess some fixed nodes (that is, nodes without any inputs) while all nodes in fixed in-degree networks change over time, meaning the scale-free networks are often effectively smaller.

All these results highlight that the critical phase transition from ordered to chaotic dynamics already occurs at an average degree of $k_c=1$ for linear networks, much sooner than for (unbiased) random networks with a phase transition at $k_c=2$~\cite{luque1997phase}, or for networks governed by canalizing functions where $k_c>2$ and depends on the specific restrictions imposed on the canalizing functions~\cite{manicka2022effective}.

\section{Theoretical Results About Dynamics of Linear Systems}\label{sec:theoretical_results}
In this section we present several theoretical results about the dynamics of linear Boolean networks. Specifically, 
in Theorem~\ref{thm:att_states} we provide a formula for the proportion of attractor states (or equivalently, for the number of periodic states).
In Theorem~\ref{thm:fp_linear} we provide a formula for the average number of fixed points for linear systems. In Lemma~\ref{lem:inv}
and Theorem~\ref{thm:inv} we provide a formula for the number of bijective linear maps. Furthermore, we provide an estimate for the percentage of linear maps that are bijective.
We also note that while this paper is about linear \emph{Boolean} networks, we derive several results for a more general case, i.e. for polynomials over a finite field with two or more values.

\subsection{Proportion of Attractor States}\label{sec:attr_states}

Here we show that if we have a Boolean network with $n$ variables and $r$ is the dimension of the nilpotent component (see Appendix A for more details of the nilpotent component),
then the proportion of attractors in the linear case is simply $\disp\frac1{2^{r}}$. Alternatively, if $r$ is the dimension of the nilpotent component, then $n-r$ is the dimension of the bijective part and therefore the number of periodic states in the linear case  $\disp2^{n-k}$.
\begin{thm}\label{thm:att_states}
The proportion of states already present in the attractors is $\displaystyle\frac1{2^r}$, where $r$ is the dimension of the nilpotent component.
\end{thm}
\begin{proof}
For every linear operator on $\mathbb F_2^n$,  there is a decomposition $\mathbb F_2^n=N\oplus B$ where $N$ is the nilpotent part and $B$ is the bijective part.
Thus every vector $v$ decomposes as $v_B+v_N$ uniquely where $v_B\in B$ and $v_N\in N$. The vectors $v=v_B+v_N$ where $v_N=0$ are already part of an attractor. Thus if $N$ has dimension $r$ over $\mathbb F_2$, then $N$ has a total of $2^r$ vectors and only one of them is zero. Hence the proportion is simply $\displaystyle\frac1{2^r}$. This completes the proof.
\end{proof}

\begin{cor}
If $r$ is the dimension of the nilpotent part, the number of periodic states (i.e., states that are part of attractors) of a linear Boolean network is $2^{n-r}$.
\end{cor}
\begin{proof}
As in the proof of Theorem~\ref{thm:att_states}, vectors $v=v_B+v_N$ where $v_N=0$ are already part of an attractor. Thus if $N$ has dimension $r$ over $\mathbb F_2$, then $n-r$ is the dimension of the bijective part $B$.
Thus, $B$ has a total of $2^{n-r}$ vectors and they are all part of attractors. Hence the number of periodic states is $\displaystyle2^{n-r}$.
\end{proof}

\subsection{Fixed Points of Linear Maps}\label{sec:fps}

\new{Even though we focus on synchronously updated linear Boolean network, we note that the results of this section do not depend on the updating schedule. That is, the average number of fixed points of linear Boolean networks will remain the same whether these systems are update synchronously or stochastically (e.g., asynchronous systems \cite{albert2015signaling}, SDDS~\cite{murrugarra2012modeling}, PBN~\cite{shmulevich2002probabilistic}).}

We now derive the expected number $\fix(M)$ of fixed points of a linear map $T(v)=Mv$ where $M\in(\bb F_q)^{n\times n}$ is the space of all $n\times n$ matrices over $\bb F_q$. Since the number of fixed points is simply the size of the null space of the shifted matrix $M-I$ it is computed by the summation
\bea{\E(\fix(M))&=&\frac1{\la(\bb F_q)^{n\times n}\ra}\sum_{M\in(\bb F_q)^{n\times n}}\#\{v:(M-I)v=0\}\nonumber\\&=&\frac1{\la(\bb F_q)^{n\times n}\ra}\sum_{M\in(\bb F_q)^{n\times n}}\#\{v:Mv=0\}.\label{expected_fixed}}
Towards the computation of the summation in \eqref{expected_fixed}, we introduce some notation. Thus the problem of counting fixed-points reduces to counting vectors in the null space.
\begin{dfn}
For any matrix $M\in(\bb F_q)^{n\times n}$ over $\bb F_q$ and $v\in(\bb F_q)^n$ a column vector, define
\be{\lt M,v\rt=\begin{cases}1&\text{if }Mv=0,\\0&\text{if }Mv\ne0.\end{cases}\label{nullsp_vect_counter}}
\end{dfn}
\begin{lem}
For all $v\in (\bb F_q)^n$, we have
\be{\sum_{M\in(\bb F_q)^{n\times n}}\lt M,v\rt=\begin{cases}q^{n^2}&\text{if }v=0,\\q^{n^2-n}&\text{if }v\ne0.\end{cases}.\label{matrices_killing_v}}
\end{lem}
\begin{proof}
From \eqref{nullsp_vect_counter} we see that the LHS of \eqref{matrices_killing_v} counts the number of matrices in whose null space $v$ lies. For every such matrix $M$, each row must lie in the orthogonal complement of $v$ using the standard inner product.

Case 1: If $v=0$, then $v$ lies in the null space, no matter the $M$, so the summand is a 1 for every matrix. The summation is therefore the total number of matrices which is $q^{n^2}$.

Case 2: If $v\ne0$, then $v$ lies in the null space of $M$, if and only if each row of $M$ is orthogonal to $v$. Thus each of the $n$ rows has $q^{n-1}$ possibilities giving $q^{n^2-n}$ possibilities.
\end{proof}

Hereafter it will be understood that $M$ runs over all matrices in $(\bb F_q)^{n\times n}$ and $v$ all vectors in $(\bb F_q)^n$. From \eqref{expected_fixed} and \eqref{nullsp_vect_counter}, since $\la(\bb F_q)^{n\times n}\ra=q^{n^2}$, we have
\be{\E(\fix(M))=\frac1{q^{n^2}}\sum_M\sum_v\lt M,v\rt\label{expected_fixed_caveman}.}

We are now in a position to compute the expected number of fixed points of a linear Boolean network.

\begin{thm}\label{thm:fp_linear}
Let $M\in(\bb F_q)^{n\times n}$. Then, \be{\E(\fix(M))=2-q^{-n}.\label{expected_fixed_exact}}
\end{thm}
\begin{proof}
From \eqref{expected_fixed_caveman} we have
\bea{\E(\fix(M))&=&\frac1{q^{n^2}}\sum_M\sum_v\lt M,v\rt=\frac1{q^{n^2}}\sum_v\sum_M\lt M,v\rt\nonumber\\&=&\frac1{q^{n^2}}\lr\sum_M\lt M,0\rt+\sum_{v\ne0}\sum_M\lt M,v\rt\rr\nonumber\\&=&\frac1{q^{n^2}}\lr q^{n^2}+\sum_{v\ne0}q^{n^2-n}\rr\quad\quad\text{from }\eqref{matrices_killing_v}\nonumber\\&=&\frac1{q^{n^2}}\lr q^{n^2}+(q^n-1)q^{n^2-n}\rr\quad\text{since }\la(\bb F_q)^n-\{0\}\ra=q^n-1,\label{E_after_summation_switch}}
which simplifies to the RHS.
\end{proof}
Since the zero vector is always a fixed point, we can conclude that on average the number of non-zero vectors that are fixed is one. 

\subsection{Expected Dimension of the Space of Fixed Points}

In this section we compute the average dimension of the fixed-point space by computing the average null space dimension. But first, we derive some preliminary results.
\begin{thm}
The number of $n\times n$ matrices with rank $d$ is 
\be{q^{d(d-1)/2}\prod_{i=0}^{d-1}\frac{(q^{n-i}-1)^2}{q^{i+1}-1}\label{matrices_with_d_rank}}\label{matrices_with_d_rank_thm}
\end{thm}
\begin{proof}
We specialise to the case $m=n$ in \cite[Theorem 13.2.5]{handbook2013finitefields} and set $k=d$.
\end{proof}

\begin{thm}
The expected dimension of the null space of a matrix $M\in(\bb F_q)^{n\times n}$ is given by
\be{
f_q(n) = \frac1{q^{n^2}}\sum_{d=0}^n(n-d)\cdot q^{d(d-1)/2}\prod_{i=0}^{d-1}\frac{(q^{n-i}-1)^2}{q^{i+1}-1}\label{expected_dimension}}
\end{thm}
\begin{proof}
Summing over the matrices of all ranks $d$ and dividing by $q^{n^2}$ - the total number of matrices - the theorem follows from Theorem \ref{matrices_with_d_rank_thm} and the fact that the respective null-space dimensions (or nullities) are $n-d$ by rank-nullity theorem.
\end{proof}

In the Boolean case ($q=2$), we have 

\begin{align}\label{tab:exp_nullsp}
    f_2(n=1) &=0.5 \nonumber \\
    f_2(n=2) &=0.6875 \nonumber \\
    f_2(n=3) &=0.771484 \nonumber  \\
    f_2(n=4) &=0.811447 \\
    f_2(n=5) &=0.830962 \nonumber  \\
    &\vdots \nonumber \\
    f_2(n=32) &=0.850179 \nonumber 
\end{align}

It can be conjectured that $1$ is closest to the average dimension of the null space. This essentially corroborates the earlier result from Theorem~\ref{thm:fp_linear}: a linear Boolean network possesses approximately $2^1=2$ fixed points, which means a $1$-dimensional space of fixed points.

\subsection{Modal Dimension of the Fixed-point Space for the Boolean Case}\label{sec:dim_fps}
In this section we quantify the proportion of $n\times n$ matrices with a $1$-dimensional null-space.

\begin{nota}
The proportion of matrices in $(\bb F_2)^{n\times n}$ with $d$-dimensional null space is denoted  by $P(d,n)$.
\end{nota}
\begin{lem}\label{lem:inv}
The proportion of invertible matrices in $(\bb F_2)^{n\times n}$ is
\be{P(0,n)=\prod_{i=1}^n(1-2^{-i})\label{invertibles'_proportion}}
\end{lem}
\begin{proof}
Setting $q=2$ and $d=n$ in \eqref{matrices_with_d_rank} we obtain
\be{\prod_{i=0}^{n-1}(2^n-2^i)\label{invertible_matrices_over_Z2}}
We have a total of $2^{n^2}$ matrices. So instead of dividing by $2^{n^2}$ at once, we divide each of the $n$ product terms in the RHS of \eqref{matrices_with_d_rank} by $2^n$ and obtain the result.
\end{proof}
\begin{lem}
For $0<x<1$, we have
\be{-\ln(1-x)=x+\frac{x^2}2+x^3\int_0^1\frac{t^2dt}{1-xt}\label{logarithm_identity}}
\end{lem}
\begin{proof}
\textit{(due to Gergő Nemes)} By division algorithm we have
\be{\frac1{1-s}=1+s+\frac{s^2}{1-s}\label{rational_function_identity}}
Integrating both sides of \eqref{rational_function_identity} from $0$ to $x$, we get
\be{-\ln(1-x)=x+\frac{x^2}2+\int_0^x\frac{s^2ds}{1-s}\label{integrate_both_sides}}
Now we do a change of variable $s=xt$, so that $ds=x\cdot dt$. The equation \eqref{integrate_both_sides} now becomes
\be{-\ln(1-x)=x+\frac{x^2}2+x^3\int_0^1\frac{t^2dt}{1-xt}}
\end{proof}
\begin{lem}
For $0<x\le\frac12$, we have
\be{-\ln(1-x)\le x+\frac{x^2}2+(8\ln 2-5)x^3\label{bound_on_ln_1-x}}
\end{lem}
\begin{proof}
From \eqref{logarithm_identity}, for $x\le\frac12$ we have
\bean{-\ln(1-x)&=&x+\frac{x^2}2+x^3\int_0^1\frac{t^2dt}{1-xt}\\&\le&x+\frac{x^2}2+x^3\int_0^1\frac{t^2dt}{1-t/2}\\&=&x+\frac{x^2}2+(8\ln 2-5)x^3}
\end{proof}
\begin{lem}
For all $n$ and $q=2$ we have \be{P(0,n)>2^{-8/7}e^{-19/42}\label{porportion_of_invertibles}}
\end{lem}
\begin{proof}
From \ref{invertibles'_proportion}, we have
\bea{-\ln P(0,n)&=&\sum_{i=1}^n\ln(1-2^{-i})\quad\quad\text{(from }\eqref{invertibles'_proportion}\nonumber\\&\le&\sum_{i=1}^n2^{-i}+\frac12\sum_{i=1}^n2^{-2i}+(8\ln2-5)\sum_{i=1}^n2^{-3i}\nonumber\\&&\text{(from }\eqref{bound_on_ln_1-x}\text{ since }q^{-i}\le\frac12)\nonumber\\&=&(1-2^{-n})+\frac{(1-2^{-2n})}6+(8\ln2-5)\frac{(1-2^{-3n})}7\nonumber\\&\le&(1-2^{-3n})\lr1+\frac16+\frac{8\ln2-5}7\rr\label{negative_log_bound}}
From \eqref{negative_log_bound}, after changing sign and exponentiating both sides we get
\be{P(0,n)\ge\lr 2^{-8/7}e^{-19/42}\rr^{1-2^{-3n}}>2^{-8/7}e^{-19/42}\label{final_P(0,n)_bound}}
\end{proof}
\begin{lem}
\be{P(1,n)=2(1-2^{-n})P(0,n)\label{proportion_relation}}
\end{lem}
\begin{proof}
Substituting $d=n-1$ and $q=2$ in \eqref{matrices_with_d_rank} we obtain
\be{\frac{2^n-1}{2-1}\cdot\prod_{i=0}^{n-2}(2^n-2^i)\label{matrices_1_dim_nullsp}}
Dividing \eqref{matrices_1_dim_nullsp}, by $2^{n^2}$ by means of dividing by $2^n$ each of the $n-1$ terms in the product as well as the fraction outside the product, we get
\be{P(1,n)=(1-2^{-n})\prod_{i=2}^n(1-2^{-i})\label{proportion_1_d_nullsp}}
Comparing \eqref{proportion_1_d_nullsp} with \eqref{invertibles'_proportion}, we obtain the result to be proved.
\end{proof}

We now prove what is evident from Table~\ref{tab:prop_1_dim_nullsp}, namely that more than $50\%$ of $n\times n$ matrices have a $1$-dimensional null space, for any $n>1$.

\begin{thm}\label{thm:inv}
For $n>1$, for more than half the matrices in $(\bb F_2)^{n\times n}$ have a $1$-dimensional null space.
\end{thm}
\begin{proof}
As a corollary of \eqref{final_P(0,n)_bound} and \eqref{proportion_relation} we have
\be{P(1,n)>2^{-1/7}(1-2^{-n})e^{-19/42}.}
One immediately calculates that $P(1,n)>0.5$ or $50\%$ for all $n\ge2$.
\end{proof}
By the proved inequality, the proportion in fact exceeds $57.16\%$ for all $n\ge7$ which is more than four in every seven. From Equation~\eqref{tab:prop_1_dim_nullsp} one can see that the asymptotic proportion of $n\times n$ matrices with a $1$-dimensional null space as $n$ gets arbitrarily large is about $57.76\%$. However, in a nutshell it is proven that the modal dimension of the null-space and hence the fixed-point space is 1.

Equation~\eqref{tab:prop_1_dim_nullsp} provides further evidence that matrices with a $1$-dimensional null space are most abundant. As $n\to\infty$, the proportion of $n\times n$ matrices with a $1$-dimensional null-space approximates $0.5776$. 
\begin{align}\label{tab:prop_1_dim_nullsp}
    P(1,n=1) &=0.5 \nonumber \\
    P(1,n=2) &=0.5625 \nonumber \\
    P(1,n=3) &=0.574219 \nonumber  \\
    P(1,n=4) &=0.576782 \\
    P(1,n=5) &=0.577383 \nonumber  \\
    &\vdots \nonumber \\
    P(1,n=32) &=0.577576 \nonumber 
\end{align}


It is thus more likely than not that the null space of a random matrix $M\in(\bb F_2)^{n\times n}$ is 1-dimensional.

\section{Discussions and Conclusions}\label{sec:discussion}

Linear Boolean networks constitute a simple class of Boolean networks and several tools have been derived for obtaining information about their dynamics. In particular, there exists a mathematical framework for obtaining the number of attractors and the height of their basins for linear Boolean systems~\cite{RHT}. Basically, one can obtain a generating function that provides the number of cycles of a given length from the minimal polynomial associated to the matrix of a linear system. We implemented the approach in~\cite{RHT} which allowed us to quantify the attractor distribution of a linear system in terms of the average number of attractors of a given length without identifying the actual states in the attractors. In this paper, we study the stability of linear Boolean networks with different in-degree distributions by calculating Derrida curves and quantifying the number of attractors and cycle lengths imposed by their network topologies. Derrida curves are commonly used to measure the stability of networks and several parameters such as the average degree $K$ and the output bias $p$ can indicate if a network is stable, critical, or unstable. For random Boolean networks there is a critical connectivity value $K_c$ such that if $K<K_c$ networks operate in the ordered regime, and if $K>K_c$ networks operate in the chaotic regime. For instance, for unbiased random functions, $K_c=2$. For networks governed by canalizing functions, a class of Boolean functions with a special structure, the phase transition occurs at an even higher average degree~\cite{manicka2022effective}. In this paper, we have shown that for linear networks, which are the least canalizing and most unstable, the phase transition from order to chaos already happens at an average in-degree of $K_c=1$.
We also show that unstable networks exhibit a large number of attractors with very long limit cycles while stable and critical networks have fewer attractors with shorter limit cycles.

Additionally, we presented theoretical results to study important properties of the dynamics of linear systems such as the expected number of fixed points which here we showed is 2. We also presented a formula for the proportion of attractor states in linear systems.
Finally, we provided a formula to quantify the number of bijective linear Boolean networks. Furthermore, we provided an estimate for the percentage of linear maps that are invertible.
\new{Even though we have focused on synchronous Boolean networks, we emphasize that the results about fixed points do not depend on the updating schedule. That is, the formulas for the average number of fixed points will remain the same when we consider linear systems with stochastic updating schemes (e.g., asynchronous systems \cite{albert2015signaling}, SDDS~\cite{murrugarra2012modeling}, PBN~\cite{shmulevich2002probabilistic}). 
Likewise, many of our results in Section~\ref{sec:theoretical_results} are valid for general multistate linear networks where the variables can take on more than two states.}

We note that one reason for the instability of linear systems is likely due the lack of canalization in linear functions. 
It has been shown that Boolean networks with canalizing rules exhibit more stability where each layer of canalization contributes additional stability~\cite{li2013boolean, kauffman2004genetic, murrugarra2011regulatory, kadelka2017influence}. Linear rules are the least canalizing type of Boolean rule~\cite{kadelka2020collectively}. Thus, it is not surprising that linear networks are mostly unstable.

\new{To conclude, we note that most published models using the Boolean network approach are nonlinear and canalizing \cite{kadelka2020meta}. However, a linear representation of a nonlinear system might offer additional tools for model analysis. In fact, any Boolean network can be represented as a linear system in a higher dimensional space \cite{shmulevich2002probabilistic}. Moreover, in recent years, the semi-tensor product representation of Boolean networks have been used for the analysis and control of Boolean networks \cite{cheng2010analysis,luo2014controllability,chen2021model}. Indeed, the semi-tensor representation is a type of linear representation. Thus, the contributions of this paper may have important implications for other classes of Boolean networks.}
Finally, linear Boolean functions tend to be highly nonlinear in the continuous generalization of Boolean functions~\cite{manicka2021biological}. For instance, the XOR rule is more nonlinear than the AND or OR rules, both of which are canalizing. A comprehensive analysis of these relationships is beyond the scope of this paper and we leave it for future work.

\section{Acknowledgements}

K.C. thanks Gergő Nemes, research fellow at Alfréd Rényi Institute of Mathematics, for his pertinent input in approximating convergent products.
D.M. was partially supported by a Collaboration grant (850896) from the Simons Foundation. C.K was partially supported by a Collaboration grant (712537) from the Simons Foundation. The contribution of R.L. was partially supported by NIH Grants 1U01EB024501-
01, 1 R01AI135128-01 and 1R01GM127909-01.

\bibliographystyle{unsrt}
\bibliography{ref_linear_BNs}
\appendix

\section{Nilpotent Trees and Cycles}
In order to find all the cycles in a Linear Boolean network, we need to know how to find the products of cycles arising from different $p$-blocks for different values of $p$. In this section, we denote by $(a,b)$ and $[a,b]$ the GCD and LCM respectively of $a$ and $b$.

\begin{lem}
The height of the nilpotent tree of a linear Boolean network is the multiplicity of the linear factor $x$ in the minimal polynomial.
\end{lem}
\begin{proof}
The nilpotent tree is composed of vectors $v$ and edges $(v,Mv)$ such that the terminal vertex is the $0$-vector. The height of the tree is the length of longest path to the zero vector. Thus we are really looking for the largest $k$ for which there is a $v$ starting from which there is a path - to the zero vector - of length $k$. In such a case $M^kv=0$ and $M^{k-1}v\ne0$. Such a $k$ is the smallest index $j$ such that $M^j$ has as its null-space the entire set of transient states. Hence $k$ is the multiplicity of the linear factor $x$ in the minimal polynomial.
\end{proof}
\begin{dfn}
The order $\ord(f)$ of an irreducible polynomial $f(x)$ over a field $F$ is the order of the $[x]$ in the multiplicative (quotient) group $\disp\frac{F[x]}{f(x)}$ where $[x]$ is the residue class of the polynomial $x$.\label{order_polynomial}
\end{dfn}
There is no known formula to compute the order of an irreducible polynomial over any field. But $\ord (f^\ell)$ can be quickly computed in terms of $\ord f(x)$ whenever $f(x)$ is irreducible. In fact we have from
\begin{thm} \emph{(\cite[Theorem 3.1]{RHT})}\label{product_cycles_two}
The product of two cycles $C_m$ and $C_n$ of lengths $m$ and $n$ respectively is given by\be{C_m\times C_n=(m,n)C_{[m,n]}}where $(m,n)$ and $[m,n]$ are the GCD and the LCM respectively, of $m,n$.
\end{thm}
The above theorem can be used effectively to multiply cycles two at a time. However, it needs to be expressed differently in order to be extended naturally to products of multiple cycles. Since $mn=(m,n)[m,n]$, we have:
\begin{cor}
\be{C_m\times C_n=\frac{mn}{[m,n]}C_{[m,n]}\label{product_cycles_two_extendable}}
\end{cor}
Inductively, one can show using \eqref{product_cycles_two_extendable} above that
\be{C_{m_1}\times C_{m_2}\times\cdots\times C_{m_n}=\frac{m_1m_2\cdots m_n}{[m_1,m_2,\ldots,m_n]}C_{[m_1,m_2,\ldots,m_n]}\label{product_cycles_mult}}

\section{Number of Attractors of Different Lengths}
For completeness here we describe the process for computing the number of attractor of different lengths. For details, see~\cite{RHT}.
We first find the number attractors of each length. This is done as follows:
\begin{enumerate}
\item Let $f$ be the minimal polynomial associated to the matrix of a linear system. For each irreducible factor $p$ of $f$, we do the following:
\begin{enumerate}
\item Compute the (finite) sequence $(b_0,b_1,\ldots)$ where $b_m=\text{rank}(p(M)^m)$ by finding $L=p(M)$ and then the row-reduced echelon form of $L^m$
\item Use \eqref{block_size_by_succ_diff} to compute $n_m(p)$, the number of $p$-blocks of multiplicity $m$, for each $m$, $1\le m\le s$ where $s$ is the multiplicity of $p$ in the characteristic polynomial
\item We compute $\ord(p)$ as per Definition \ref{order_polynomial}.
\item For each $p$-block we use formula in \cite[Theorem 5]{RHT} to compute the number of cycles of each length.
\end{enumerate}
\item We then compute the number of attractors of all possible lengths taking all possible products of cycles, one each from every $p$-block for every $p$ using \eqref{product_cycles_mult}.
\end{enumerate}
\subsection{Representing All Attractors Corresponding to Each Irreducible Factor}
\begin{enumerate}
    \item For each irreducible factor $p$ of the characteristic polynomial $f$ we do the following
    \begin{enumerate}
        \item For each $\ell$, ($1\le\ell\le m$) where $m$ is the multiplicity of $p$ in the minimal polynomial of $M$, We do the following:
        \begin{enumerate}
            \item We find a set $V_i$ of vectors that satisfy the pair (say $C_\ell$) of conditions below:
        \[p(M)^{\ell-1}(v)\ne0\quad\quad p(M)^\ell(v)=0\]
        such that all the vectors found above span the null space of $p(M)^m$. This is done inductively by finding a basis of the null space of $p(M)^\ell$ and removing those vectors in that are already in the null space of $p(M)^{\ell-1}$.
        \item Additionally we find representatives of all orbits of the multiplicative group $\bb F_2[x]/(p(x)^\ell)$ under the action of multiplication by $[x]$.
        \item With the information found above, one can evaluate the matrix $M$ at the representative identified above for each orbit, then apply the resulting matrix to each vector satisfying the pair $C_\ell$ of conditions.
        \end{enumerate}
    \end{enumerate}
\item Once the attractors corresponding to each irreducible component are identified, we can find all attractors, by choosing all possible combinations choosing at most one vector (in the list generated above) from each irreducible invariant subspace. Mild usage of trial and error would be needed since a vector in the null space of $p(M)^{\ell_1}$ and another vector in the null space of $p(M)^{\ell_2}$ could lie in the same irreducible invariant subspace, which would result in a redundancy in generating the cycles.
\end{enumerate}

\section{Linear Algebra}
Some preliminaries from linear algebra are required to describe the method used to compute the number of attractors of various lengths.
\begin{thm}
Suppose matrix is a $M$ (over any field) whose characteristic and minimal polynomial are both $p(x)^k$ for some $k\ge1$ where $p$ is an irreducible polynomial of degree $d$. Then the nullity $n(p(M))$, of $p(M)$, equals $d$.
\label{frobenius_block_nullity}
\end{thm}
\begin{proof}
We can reduce $M$ to the Frobenius normal form
\be{F_{p,k}=\bpm{C&0&0&\cdots&0&0\\U&C&0&\cdots&0&0\\0&U&C&\cdots&0&0\\\vdots&\vdots&\vdots&\ddots&\vdots&\vdots\\0&0&0&\cdots&C&0\\0&0&0&\cdots&U&C}\label{frobenius}}here $C$, a $d\times d$ matrix, is the companion matrix for the polynomial $p(x)$, $U$ is the $d\times d$ matrix whose only non-zero entry is a 1 at the top right corner, and $F$ is a $dk\times dk$ matrix. The above form can be obtained by means of the ordered basis
\[v,Mv,M^2v,\ldots,M^{d-1}v,p(M)v,Mp(M)v,M^2p(M)v,\ldots M^{d-1}p(M)v,\ldots\]
\[\ldots p(M)^{k-1}v,Mp(M)^{k-1}v,M^2p(M)^{k-1}v\ldots,M^{d-1}p(M)^{k-1}v\]
where $v$ is a vector such that the smallest $j$ for which $p(M)^jv=0$ is $j=k$.

Thus $M=PDP^{-1}$ where $P$ is the $dk\times dk$ matrix whose columns are the above $dk$ basis vectors in that order. For each irreducible factor $p$ we compute the ranks of $p(M)$, $p(M)^2$, ... etc to obtain the no. of blocks corresponding to $p$, of different sizes. As a result we have\bea{p(M)&=&P\cdot p(D)\cdot P^{-1}\label{similarity}\\p(M)\cdot P&=&P\cdot p(D)\label{change_basis_matrix}}
By definition of $P$, the LHS of \eqref{change_basis_matrix} consists of all columns between the $(d+1)^{th}$ and the $dk^{th}$ column inclusive, followed by $d$ columns which are entirely $0$. This is because upon application of $p(M)$, the basis elements (respectively) yield
\[p(M)v,\ldots M^{d-1}p(M)v,p(M)^2v,\ldots,M^{d-1}p(M)^2v,p(M)^3v,\ldots,M^{d-1}p(M)^3v\ldots,\]
\[\ldots p(M)^kv,Mp(M)^kv,M^2p(M)^kv\ldots,M^{d-1}p(M)^kv.\]
However the last $d$ vectors are 0 since $p(M)^kv=0$ by the minimality of $p(x)^k$. Thus we are left-shifting the columns of $P$ by $d$. By \eqref{change_basis_matrix}, RHS also must contain the same (left-shifted) columns of $P$, which is equivalent to right-multiplication by
\be{p(D)=\bpm{0&0&0&\cdots&0&0\\I&0&0&\cdots&0&0\\0&I&0&\cdots&0&0\\\vdots&\vdots&\vdots&\ddots&\vdots&\vdots\\0&0&0&\cdots&0&0\\0&0&0&\cdots&I&0}\label{rank_first_power}}
where $I$ is the $d\times d$ identity matrix. By the similarity of $p(M)$ and $p(D)$, they must have the same rank, and the latter is at once seen to have nullity $d$.
\end{proof}
\begin{dfn}
Suppose we have an $n\times n$ matrix $M$ with characteristic polynomial $f(x)$ where $p(x)\mid f(x)$ is an irreducible factor. Define a $p$-\em block of multiplicity \em $m$ to be an irreducible invariant (under $M$) subspace $W\subset\bb R^m$ such that the minimal polynomial of $M$ restricted to $W$ is $p(x)^m$.
\end{dfn}
\begin{dfn}
Given a sequence $a_0,a_1,...,a_k$ we define the following.\begin{enumerate}
    \item[]The first order forward difference by $\Delta a_i=a_{i+1}-a_i$.
    \item[]The first order backward difference by $\nabla a_i=a_i-a_{i-1}$,
    \item[]The second order forward difference by $\Delta^2 a_i=\Delta a_{i+1}-\Delta a_i$.
    \item[]The second order backward difference by $\nabla^2 a_i=\nabla a_{i+1}-\nabla a_i$.
    \label{succ_diff}
\end{enumerate}  
\end{dfn}
\begin{lem}
\[\nabla a_i=\Delta a_{i-1}\]\label{delta_napla}
\end{lem}
\begin{cor}
Let $M$ be an $n\times n$ matrix with characteristic polyomial $f(x)$. Let $p(x)$ be an irreducible factor of $f(x)$. Define sequences $(a_r)$ and $(b_r)$ by $a_r=n(p(M)^r)$ and $b_r=\text{rank}(p(M)^r)$. Then $n_m(p)$ the number of $p$-blocks of multiplicity $m$ is
\be{n_m(p)=-\frac{\Delta^2a_{m-1}}d=\frac{\Delta^2b_{m-1}}d\label{block_size_by_succ_diff}}
\label{cor_succ_diff}
\end{cor}
\begin{proof}
By the existence of a Frobenius decomposition and Theorem \ref{frobenius_block_nullity}, the result of left-multiplication of each $p$-block of multiplicity ${m-1}$ or less, by $p(M)^{m-1}$ is the zero matrix. Now, by \eqref{similarity} and \eqref{rank_first_power}, further multiplying $p(M)$, increases by $d$ the nullity of each $p$-block of multiplicity at least $m$, leaving the remaining blocks unchanged. The nullities of the $q$-blocks $(q\ne p$) too remain unchanged since such $q$ are relatively prime to $p$. In other words,
\bea{d\sum_{r\ge m}n_r(p)&=&n(p(M)^m)-n(p(M)^{m-1})\label{block_size_by_nullity}\\
&=&\text{rank}(p(M)^{m-1})-\text{rank}(p(M)^m)\label{block_size_by_rank}}
where \eqref{block_size_by_rank} follows from rank-nullity theorem. By definition of backward differences we have
\be{d\sum_{r\ge m}n_r(p)=\nabla a_m=-\nabla b_m\label{block_size_by_rank_diff}}
From \eqref{block_size_by_rank_diff} we have
\be{d\lr\sum_{r\ge m}n_r(p)-\sum_{r\ge m+1} n_r(p)\rr=d\cdot n_m(p)=\nabla a_m-\nabla a_{m+1}=\nabla b_{m+1}-\nabla b_m\label{block_size_by_rank_bdiff}}
By Lemma \ref{delta_napla} we have
\bea{&&d\cdot n_m(p)=\Delta a_{m-1}-\Delta a_m=\Delta b_m-\Delta b_{m-1}\nonumber\\&&d\cdot n_m(p)=-\Delta^2a_{m-1}=\Delta^2b_{m-1}\label{block_size_by_rank_fdiff}}
and the result follows.
\end{proof}

\end{document}